\documentclass[10pt, journal]{IEEEtran}

\ifCLASSINFOpdf
\else
\fi
\usepackage{multicol}
\usepackage{lipsum}
\usepackage{setspace}
\usepackage{amsmath}
\usepackage{cases}
\usepackage{cite}
\usepackage{amsfonts}
\usepackage{amssymb}
\usepackage[linesnumbered,ruled,vlined]{algorithm2e}

\usepackage{array}
\usepackage[margin=1.0in]{geometry}
\usepackage{amsthm}
\usepackage{empheq}
\renewcommand*{\stackrel}{%
\mathrel\bgroup\stack@relbin
}

\newtheorem{theorem}{Theorem}[section]
\newtheorem{lemma}[theorem]{Lemma}
\newtheorem{proposition}[theorem]{Proposition}

\ifCLASSINFOpdf
   \usepackage[pdftex]{graphicx}
  \graphicspath{{../pdf/}{../jpeg/}{../eps/}}
  \DeclareGraphicsExtensions{.eps,.pdf,.jpeg,.png}
\else
   \usepackage[dvips]{graphicx}
   \graphicspath{{../eps/}}
   \DeclareGraphicsExtensions{.eps}
\fi
\usepackage{epsfig}
\usepackage{fmtcount}
\usepackage{subfigure}
\usepackage{array}
\usepackage{url}
\usepackage[linesnumbered,ruled,vlined]{algorithm2e}
\usepackage{multicol}
\usepackage{float}
\usepackage{lipsum}
\usepackage{color}

\usepackage{mathtools}  
\usepackage{tabulary}
\usepackage{booktabs}
\usepackage{stfloats}

\begin{document}

\title{{\fontsize{16}{16}\selectfont A Waterfilling Algorithm for Multiple Access Point Connectivity with Constrained Backhaul Network}}

\author{Syed~Amaar~Ahmad
\thanks{
Syed A. Ahmad received PhD in Electrical Engineering from Virginia Tech. He is currently working as an R{\&}D Engineer with Savari Networks. Email: saahmad@vt.edu.
}}

\markboth{IEEE Wireless Communication Letters,~Vol.~, No.~, ~}%
{Shell \MakeLowercase{\textit{et al.}}: Bare Demo of IEEEtran.cls for Journals}

\maketitle

\begin{abstract}
The 4G/5G paradigm offers User Equipment (UE) simultaneous connectivity to a plurality of wireless Access Points (APs). We consider a UE communicating with its destination through multiple uplinks operating on orthogonal wireless channels with unequal bandwidth. The APs are connected via a tree-like backhaul network with destination at the root and non-ideal link capacities. We develop a power allocation scheme that achieves a near-optimal rate without explicit knowledge of the backhaul network topology at transmitter side. The proposed algorithm \emph{waterfills} a dynamic subset of uplinks using a low-overhead backhaul load feedback.
\end{abstract}


\IEEEpeerreviewmaketitle

\section{Introduction}
\IEEEPARstart{T}{he} paradigm of device-centric architecture promises to let wireless devices transmit information flows via several heterogeneous nodes \cite{5G}. Dense smalls cells established with a plurality of wireless Access Points (APs) and simultaneous multi-tier communication (LTE and WiFi) will enable devices to have connectivity to multiple access points. The  backhaul of APs is increasingly being viewed as a performance-limiting factor \cite{Chia_backhaul_2013}.

We consider here a single UE streaming data to its destination via its wireless uplinks to APs operating on orthogonal channels and cascaded to a complex backhaul network. {Waterfilling} is a well-known allocation that achieves capacity for a multi-band system \cite{PerezIEEESigProc2005}. When the backhaul capacities are constrained, such as when numerous internal data flows are sharing a backhaul, then such an allocation is sub-optimal. Unlike \cite{CuiCISS2014} we propose a waterfilling algorithm for a tree-structured backhaul network with destination at the root (see Fig.~\ref{fig:furt1}) and where the transmitter does not require explicit knowledge of the backhaul link capacities. Unlike \cite{WeiYuBackhaul2014} that considers one-hop backhaul links, a general tree topology may be deployed for better scalability and latency performance \cite{Tree2010}. Under the algorithm, the UE waterfills a dynamic subset of its uplinks with sufficient backhaul capacity and reduces transmit power for those which are constrained using a novel low-overhead feedback on the effective backhaul load.

\begin{figure}[!t]
\begin{center}
        \centering\includegraphics[width=0.46\textwidth]{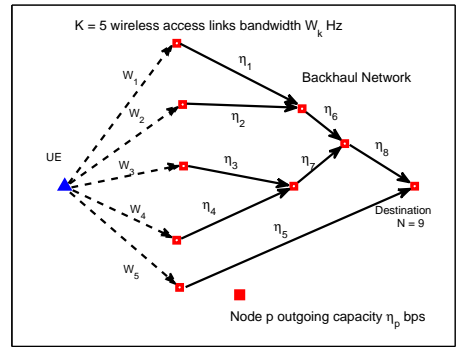}
\caption{A depiction of a UE connected to multiple APs sharing a backhaul network with a tree hierarchy.}
\label{fig:furt1}
\end{center}
\end{figure}
\section{System Model}
We consider a UE with wireless uplinks to $K$ APs. The APs are part of a backhaul network that is denoted by a set ${\mathcal{N}}=\{1,2,\cdots, K, \cdots, N\}$, where $N$ represents the destination node and node $p:K<p<N$ is an intermediate backhaul node. The $K$ uplinks operate over a set of orthogonal narrowband channels with corresponding bandwidths given as ${\mathcal{W}}=[W_1, W_2, \cdots, W_K]$ Hz with the total bandwidth as $W = \sum^K_{k=1} W_k$. 
The backhaul network is arranged as a tree structure with a root node (destination), intermediate nodes and end nodes (APs). Each child node has only one parent such that all branches originating from the APs converge to destination node $N$ (see Fig. \ref{fig:furt1}). The tree topology is denoted by defining the parent of each backhaul node as a set ${\mathcal{H}}$, where node $p$ forwards data to its parent node $j$ such that ${\mathcal{H}}(p) = j$ where $K < j\leq N$. 

The UE's available transmit power is given as ${{P}}_{\max}$ Watts. It transmits to AP $k$ with power $P_k$ where $1\leq k \leq K$ where the transmit power vector is denoted by ${\bf{P}}=[P_1,P_2\cdots, P_K]$. Given $g_{k}$ as the channel power gain and $n_k$ as corresponding noise power, the effective noise level $E_k$, signal-to-noise ratio $\gamma_k$ and rate $\eta_k$ on uplink to AP $k$ are:
\begin{equation}\label{sinr}
\begin{split}
& E_k = \frac{n_k}{g_k}, \\
& \gamma_k = \frac{P_{k}}{E_{k}},\\
\end{split}
\end{equation}
\begin{equation*}
\eta_k = W_k\log_2(1+\gamma_k).
\end{equation*}
The capacity between backhaul node $p: p>K$ to its parent (i.e. ${\mathcal{H}}(p)$) is denoted as $\eta_p$. The achievable rate at any node $p$ in the backhaul network can be recursively defined as follows:
\begin{equation}\label{Recursive_Rate}
\begin{split}
R_p=\left\{
                \begin{array}{ll}
                \sum_{\forall j: {\mathcal{H}}(j)=p}R_j, &\mbox{   $p = N$}\\
                  \min\left(\eta_p, \sum_{\forall j: {\mathcal{H}}(j)=p}R_j\right), &\mbox{ $K<p<N$}\\
                  \min\left(\eta_p, W_p\log_2(1+\gamma_p)\right), &\mbox{ $p\leq K$}                  
                \end{array}
              \right.
\end{split}
\end{equation}
where $R_j$ is the achievable rate at node $j$ that is being forwarded to parent $p$. Eq. \eqref{Recursive_Rate} is based on the Max-Cut Min-Flow Theorem \cite{sherali} where at node $p$, the achievable rate $R_p$ is simply the minimum of the aggregate incoming rate from child nodes and the outgoing rate.

\section{Problem Formulation}
The resource allocation maximization of the achieved rate $R_N$ in \eqref{Recursive_Rate} at destination node $N$ is: 
\begin{equation}
\begin{split}\label{opt_func_het}
R^* = \underbrace{\max}_{{\bf{P}}} \hspace{0.1in} & R_N\left({\bf{P}}\right) \\
\sum^K_{k=1} P_{k} & \leq {P_{\max}}\\
P_{k} &\geq 0.
\end{split}
\end{equation}
\begin{proposition}
If the backhaul capacity of all nodes is high enough, $P_{k} = (W_k\mu - E_k)^+$ is an optimal allocation for  \eqref{opt_func_het} with $\mu$ as the waterfilling level for all uplinks. \label{WF_het} 
\end{proposition}
\begin{proof}
We modify the results \cite[pp~5]{PerezIEEESigProc2005} for unequal bandwidth. For a waterfilling level $\mu$, channel $k$ of bandwidth $W_k$ must be allocated $P_{k}= \left(\left(\mu - \frac{E_k}{W_k}\right)W_k\right)^+=(W_k\mu - E_k)^+$ where $\frac{E_k}{W_k}$ is the channel's noise power spectral density. Algorithm~\ref{algo:WFLevel} illustrates the determination of $\mu$.
\end{proof}
To implement \eqref{opt_func_het} for general cases the UE would require complete knowledge of the backhaul network parameters as a single water-filling level would not work for uplinks with over-loaded backhaul nodes. We assume that channel state information (CSI) for achievable rate on each of the $K$ uplinks, as denoted by \eqref{sinr}, is available to the UE whereas explicit knowledge of the backhaul network is not. We propose a simple description to characterize the end-to-end backhaul capacity from any AP to the destination node. We define the \emph{rate differential} as the difference at any point $p$ in the backhaul hierarchy between its outgoing capacity and the aggregate incoming data rate:
\begin{equation}\label{Vr_het}
\begin{split}
V_p=\left\{
                \begin{array}{ll}
                  \eta_p - W_p\log_2(1+\gamma_p),  &\mbox{ $p\leq K$}\\
                  \eta_p - \sum_{\forall j: {\mathcal{H}}(j)=p} R_j, &\mbox{ $K+1\leq p$.}
                \end{array}
              \right.
\end{split}
\end{equation}
Whenever $V_p < 0$, the outgoing link from node $p$ is a \emph{bottleneck} link. A rate mismatch between the uplink and a backhaul link creates a congestion effect at the intermediate node \cite{Neely2005} and where the UE can simultaneously reduce its transmit power and rate. Moreover, this may instead be re-allocate power to other uplinks. Conversely, $V_p\geq 0$ indicates that the backhaul capacity via $p$ is high enough to support the current load. Thus, $V^+_p=\min\left(V_p,0\right)$ denotes the achievable rate improvement.  

We assume that a positive-valued constant $\tau$ is a rate differential threshold that represents some tolerable load at backhaul node $p$. We can denote the backhaul capacity state at node $p$ in Table~\ref{tab:TableOfNotationForMyResearch}.
\begin{table}[t]\caption{Node $p$'s Backhaul State.}
\begin{center}
\begin{tabular}{r c c p{0.1cm} }
\toprule
State ${\mathcal{S}}_p$& Rate Differentials& Description\\
\bottomrule
\small
$1$ & $ V_p\geq 0 $ & \mbox{Under-utilized}\\ 
$2$ &$-\tau\leq V_p<0 $ & \mbox{Balanced}\\
$3$ & $V_p<-\tau $& \mbox{Over-loaded}\\
\bottomrule
\end{tabular}
\end{center}
\label{tab:TableOfNotationForMyResearch}
\end{table}
Note that each of the 3 possible backhaul states on the $K$ APs (i.e. (i) $V_p(t)\leq 0$, (ii) $V_p(t)<-\tau$ and (iii) $-\tau\leq V_p(t)<0$) can be represented as $\lceil{\log_2(3)}\rceil= 2$ bits. The effective backhaul state ${\mathcal{S}}^{(k)}_{eff}$ of the end-to-end path from AP $k: k\leq K$ to the destination is recursively computed:
\begin{equation}
\begin{split}\label{RecurState}
{\mathcal{S}}^{(k)}_{eff} &= \max(\max(\max\cdots\max\left({\mathcal{S}}_p, {\mathcal{S}}_j) \cdots{\mathcal{S}}_i),{\mathcal{S}}_N\right)\\
&=\max\left({\mathcal{S}}_p, {\mathcal{S}}_j,\cdots, {\mathcal{S}}_i,{\mathcal{S}}_N\right).
\end{split}
\end{equation}
Note that $p= {\mathcal{H}}(k)$, $j = {\mathcal{H}}(p)$, $N = {\mathcal{H}}(i)$ and so on are nodes along the path from AP $k$ to destination $N$. Basically, in \eqref{RecurState}, from every AP the effective state to a backhaul node along the route is propagated as $2$ bits until this information reaches the destination. The destination can then compress load information from all $K$ APs into $\lceil{\log_2(3^K)}\rceil= \lceil{K\log_2(3)}\rceil$ feedback bits for the UE. In practice the backhaul states may be computed for an interval spanning several LTE resource blocks comprising several hundred bits. This negligible feedback can then be embedded on a downlink wireless control channel common to all APs \cite{3gppDCP12}. In contrast, to directly implement the optimization in \eqref{opt_func_het}, the information overhead required by the UE will scale as ${\mathcal{O}}\left((N)\lceil\log_2(N)\rceil\right)$ where $\lceil\log_2(N)\rceil$ bits are needed to uniquely identify the parent of each of the $N$ nodes (where $N>K$) with additional bits for representing its backhaul capacity.

\section{Waterfilling with Load Feedback}
We assume that adaptations occur in time intervals denoted as $t \in \{0,1,2,..\}$. The proposed transmit power algorithm basically works as follows: Initially, at $t=0$, the UE uses classic waterfilling allocation on all uplinks using maximum transmit power. Once it has acquired the corresponding backhaul load information ($t>0$), it maintains its transmit power for some AP $k$ that is load-balanced (i.e. ${\mathcal{S}}^{(k)}_{eff}(t)=2$) and reduces it by a scaling factor $Z: 0<Z<1$ for those that are over-loaded (i.e. ${\mathcal{S}}^{(k)}_{eff}(t)=3$). For the remaining (\emph{residual}) transmit power, a new waterfilling level is then determined and power allocation is made for the remaining wireless links with sufficient effective backhaul capacity. The algorithm thus seeks to incrementally raise the waterfilling level of those uplinks whose effective backhaul capacities are large enough. 
\begin{algorithm}[t]
{  
  \If{$t=0$}{
   $P_{r}(t+1)=  P_{\max}$\;
   }
   \Else{
   	$P_{r}(t+1)=  P_{\max}-\sum^K_{k=1: {\mathcal{S}}^{(k)}_{eff}(t)=2} P_k(t)-\sum^K_{k=1: {\mathcal{S}}^{(k)}_{eff}(t)=3} Z.P_k(t)$\;
   }
$t=t+1$\;   
Reset $t=0$ if $\exists k: E_k(t) \neq E_k(t-1)$\;  
}
\caption{Residual transmit power.}\label{algo:residualPower}
\end{algorithm}
\begin{algorithm}[t]
{Initialization: {$\mu(t)=0$, $\overset{\sim}{\mu}=0, \overset{\sim}{K} = K$}\\
Sort pairs $(\overset{\sim}{E_j}, \overset{\sim}{W_j}): \overset{\sim}{E_j}/\overset{\sim}{W_j} \leq \overset{\sim}{E_{j+1}}/\overset{\sim}{W_j}, \forall j$\;
 
  \While{$\mu(t)=0$}{
   $\overset{\sim}{\mu} = \frac{P_r(t+1) + \sum_{j=1}^{\overset{\sim}{K}} \overset{\sim}{E_j}}{ \sum^{\overset{\sim}{K}}_{j=1} \overset{\sim}{W_j}}$\;
   
   \If{$\overset{\sim}{\mu} \geq \max(\frac{\overset{\sim}{E_{1}}}{\overset{\sim}{W_{1}}}, \cdots, \frac{\overset{\sim}{E_{\overset{\sim}{K}}}}{\overset{\sim}{W_{\overset{\sim}{K}}}})$}{
   $\mu(t+1) = \overset{\sim}{\mu}$\;
   }
   \Else{
   $\overset{\sim}{K} = K-1$\;
   }
 } 
}\caption{Updated waterfilling level $\mu(t+1)$.}\label{algo:WFLevel}
\end{algorithm}
Algorithm~\ref{algo:residualPower} describes the computation of {residual} transmit power where the last condition denotes a reset due to time-varying channel gain changes induced by fading. Algorithm~\ref{algo:WFLevel} then describes how a  waterfilling level is then determined for the given residual power. It works by sorting uplinks by effective noise levels and eliminating those channels for any allocation which suffer from high noise levels. Note that equation on line $4$ in Algorithm~\ref{algo:WFLevel} is derived from the fact that the residual transmit power equals the allocation to the $\footnotesize\overset{\sim}{K}$ channels where $\footnotesize P_r(t+1) = \sum_{j=1}^{\overset{\sim}{K}}(\mu - \overset{\sim}{E_j}/\overset{\sim}{W_j})\overset{\sim}{W_j}$. Finally, the transmit power is then updated as follows:
\begin{subequations}\label{power_alloc_het}
\small
\begin{empheq}[left={{P_k(t+1)=}}\empheqlbrace]{align}
&(W_k\mu(t+1)- E_k(t))^+, \mbox{ ${\mathcal{S}}^{(k)}_{eff}(t)=1$}\label{pa_het}\\
&P_k(t), \hspace{0.95in}\mbox{ ${\mathcal{S}}^{(k)}_{eff}(t)=2$}\label{pb_het} \\
&Z.P_{k}(t), \hspace{0.82in}\mbox{ ${\mathcal{S}}^{(k)}_{eff}(t)=3$}. \label{pc_het}
\end{empheq}
\end{subequations}
For the uplinks in state ${\mathcal{S}}^{(k)}_{eff}(t)=1$, the UE allocates it residual transmit power using the waterfilling allocation as per Proposition~\ref{WF_het}. Otherwise, if it is in ${\mathcal{S}}^{(k)}_{eff}(t)=2$, the effective backhaul capacity for the uplink is balanced and hence the transmit power is maintained in \eqref{pb_het}. Finally in ${\mathcal{S}}^{(k)}_{eff}(t)=3$, since the backhaul link is over-loaded the UE reduces its transmit power as per \eqref{pc_het} by $Z$. 
If $Z$ is set to a low value there may be too fast a reduction in data rates between successive iterations causing an over-loaded backhaul link to change to an under-utilized link and vice versa thereby causing an unstable achieved rate.
\begin{lemma}
The decrease in the achievable rate between successive iterations is upper bounded such that $R_N(t)-R_N(t+1)\leq - \max(W_1, W_2, \cdots, W_K)\log_2(Z)$. \label{MaxRateCH}
\end{lemma}
\begin{proof}
Suppose that at AP $k$ current rate is $R_k(t) = W_k\log_2(1+\gamma_k(t))$. If ${\mathcal{S}}^{(k)}_{eff}(t)=3$ then the UE updates this to $R_k(t+1) = W_k\log_2(1+Z.\gamma_k(t))$. Thus $R_k(t)- R_k(t+1)\leq W_k(\log_2(1+\gamma_k(t))-\log_2(1+Z.\gamma_k(t)))<W_k \log_2((1+\gamma_k(t))/(1+Z.\gamma_k(t)))\leq -W_k\log_2(Z)\leq - \max(W_1, W_2, \cdots, W_k)\log_2(Z)$.
\end{proof}
\begin{proposition}
If $2^{\frac{-\tau}{\max(W_1,W_2,\cdots, W_K)}}< Z<1 $, then the achieved rate converges to an equilibrium where it is within a target $\tau$ bps of the optimum rate $R^*$ in \eqref{opt_func_het}. 
\label{Equilib}
\end{proposition}
\begin{proof}
Note that transmit power (and link rate) on an uplink can be reduced by some scaling factor $Z$ such that no backhaul node transitions from State 1 to state 3 or vice versa. If $Z>2^{\frac{-\tau}{\max(W_1, \cdots, W_k)}}$ or $-\max(W_1,\cdots, W_k)\log_2(Z)<\tau$ equivalently then as per Lemma~\ref{MaxRateCH} even the largest rate decrease along any end-to-end path will be less than $\tau$ between consecutive iterations. Thus ${\mathcal{S}}^{(k)}_{eff}(t)=3$ would not go to ${\mathcal{S}}^{(k)}_{eff}(t+1)=1$. Thus, no backhaul link suddenly goes from being over-loaded to under-utilized or vice versa. Eventually, all over-loaded backhaul nodes reach ${\mathcal{S}}^{(k)}_{eff}(t)=2$ (become load-balanced) at some $t$. It implies that the allocations and achieved rate eventually converge. Next consider the optimality gap. Assume an initially unconstrained backhaul network where the UE maximizes rate by simply using classic waterfilling based on CSI of the $K$ uplinks. Now suppose that the backhaul capacities are constrained for some uplinks along the path to destination. To solve the optimization \eqref{opt_func_het}, the transmission rate can be reduced exactly just enough to match the constrained capacity along these paths and also correspondingly decrease transmit power on these uplinks. This saved power can then be re-allocated to uplinks with sufficient backhaul capacity to increase their waterfilling level so as to improve the achieved rate. Our scheme works the same way where the transmission power and rate, in \eqref{sinr}, is decreased for those uplinks with constrained backhaul links. Despite limited knowledge, our power decrease still achieves a rate equal to the constrained capacity along these paths. The waterfilling level for remaining uplinks, with backhaul nodes in State 1, is then increased by allocating them this additional residual power as per Algorithm~\ref{algo:WFLevel}. Substituting $Z = 2^{\frac{-\tau}{\max(W_1, \cdots, W_k)}}$ in the inequality in Lemma~\ref{MaxRateCH} implies that the achievable rate increase for these uplinks is upper-bounded such that $R_N(t+1)-R_N(t)\leq\tau$ bps. This in turn implies that our scheme, with limited knowledge, achieves a rate within $\tau$ bps of the optimal rate.
\end{proof}
The above intuitively indicates that if $\tau$ is close to zero (i.e. small gap between the achieved rate and the optimal $R^*$) then this comes at the cost of slow convergence.
\begin{figure*}[!t]
\begin{center}
\subfigure[]{\includegraphics[width=0.49\textwidth]{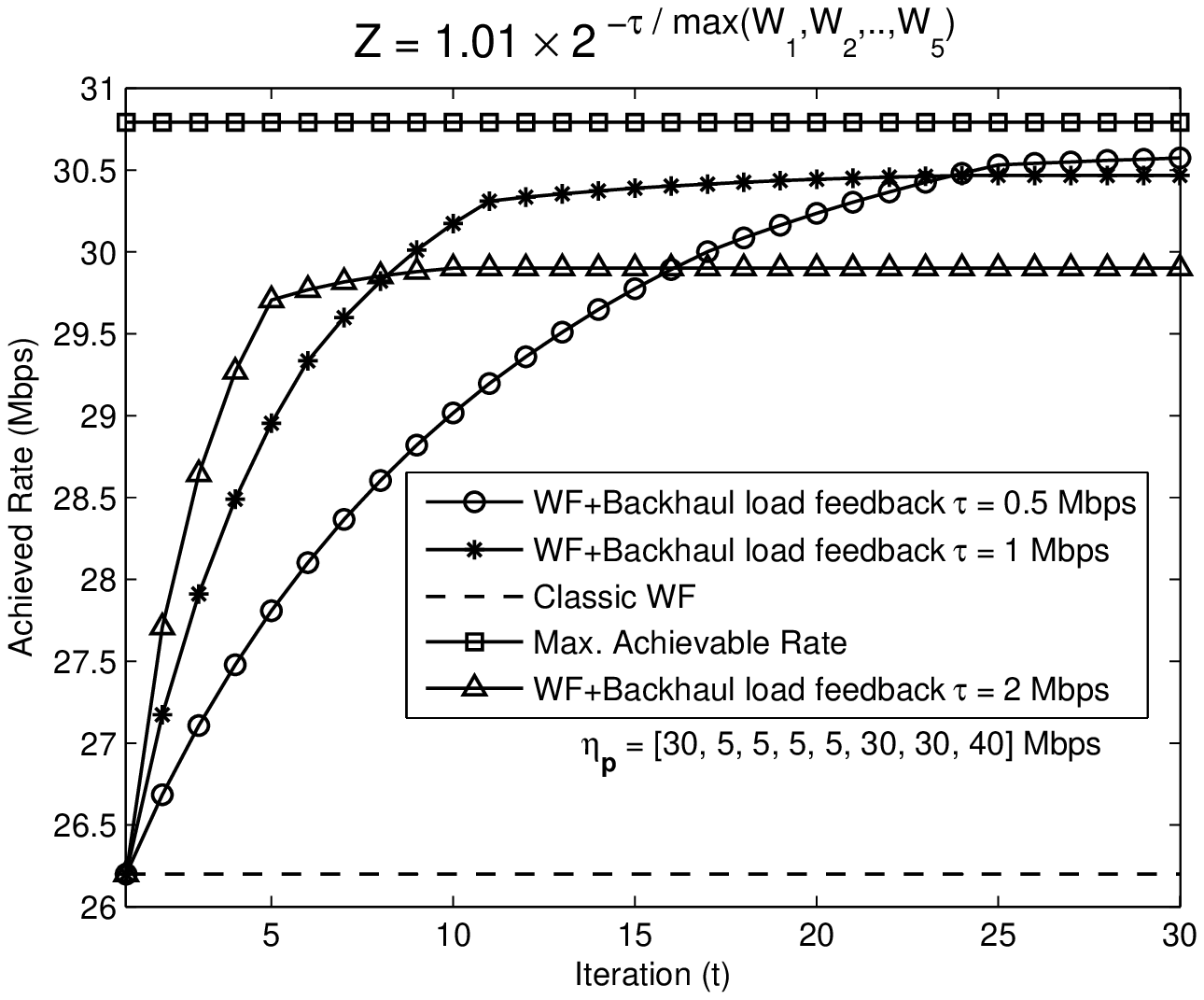}}
\subfigure[]{\includegraphics[width=0.49\textwidth]{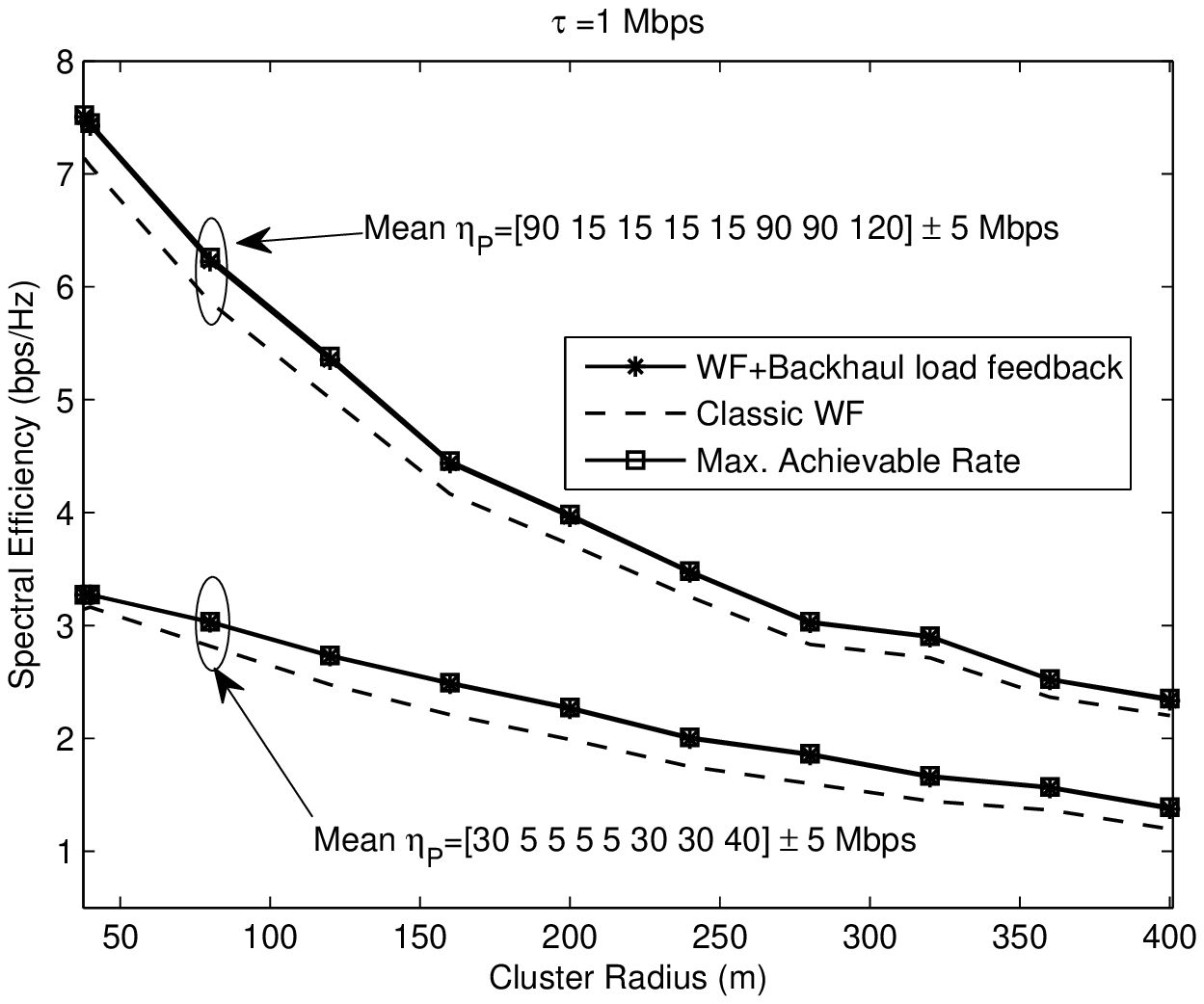}}
\caption{Under the proposed scheme: (a) as per Proposition~\ref{Equilib}, the achieved rate iteratively evolves to within a gap of $\tau$ bps around maximum $R^*$ with a tradeoff with convergence speed, and (b) system achieves a near-optimal performance over the range of cell sizes.}
\label{fig:n}
\end{center}
\end{figure*}

\section{Simulation Results}
In this section, we use simulations to demonstrate the scheme's performance. We assume the noise power spectral density as $-190$ dBW/Hz and $P_{\max}=1.0$ Watts. We assume that the power gains are of the form $g_k = \kappa_k\cdot d_{k}^{-\alpha}$ where $d_{k}$ is distance between UE and AP $k$, path loss exponent $\alpha = 4.0$ and $\kappa_k$ is a lognormally distributed shadowing gain with $5$ dB variance. The shadowing gains are assumed independent and identically distributed for all channels. The channel bandwidth per uplink is set at $1$ MHz, $2$ MHz or $5$ MHz with equal probability. The APs are located randomly and uniformly within a $D$ m cluster radius centered around the UE. We use the backhaul network depicted in Fig. \ref{fig:furt1} with $K=5$. We compare the proposed scheme (WF with backhaul load feedback) with: (i)  optimization of \eqref{opt_func_het} (i.e. maximum achievable rate), and (ii) classic waterfilling (WF) for all its uplinks. We set the value $Z = 1.01 \times 2^{\frac{-\tau}{\max(W_1, \cdots, W_K)}}$ in all trials.

We let ${\bf{\eta}}_P=[\eta_1, \cdots, \eta_{N-1}]$ denote an instant vector of backhaul link capacities with corresponding values shown in the figure. We first plot the achieved rate in Fig. \ref{fig:n}a for a sample network over time. In this case we set effective noise vector as ${\bf{E}}=[E_1 \cdots E_k] = [10^{-2}, 10^{-1}, 10^{-4}, 10^{-1}, 10^{-3}]$ Watts and ${\mathcal{W}} = [2, 5, 1, 2, 1]$ MHz. We observe that the scheme significantly outperforms classic waterfilling and iteratively approaches the optimum $R^*$ when $\tau$ is set to a lower value. In Fig. \ref{fig:n}b we plot the achieved rate per unit bandwidth (i.e. spectral efficiency) by varying the cell's cluster radius. We also randomize the capacity of every backhaul link using uniform distribution with a range of $\pm 5$ Mbps around a mean ${{\bf{\eta}}_P}=[{\overline{\eta}}_1, \cdots, {\overline{\eta}}_{N-1}]$. We again observe that the proposed scheme approaches the optimum value for both the heavily and the lightly constrained backhaul regimes. 

\section{Conclusion}
A transmit power allocation scheme for a UE with multiple access point connectivity with non-ideal backhaul network has been proposed. It has been demonstrated that even without explicit backhaul network knowledge, the UE can approach optimal rate to within a target bound of $\tau$ bps. Future work could consider the impact of multiple users and effects of interference.

\bibliography{Citations}

\end{document}